\documentclass[11pt]{article}
\usepackage[a4paper,margin=25mm]{geometry}
\usepackage[T1]{fontenc}
\usepackage{lmodern}
\usepackage{amsmath,amssymb,amsthm,booktabs,microtype}
\usepackage[hidelinks]{hyperref}
\hypersetup{pdftitle={MacKay--Neal Codes Achieve Capacity under MAP Decoding},pdfauthor={Kenta Kasai}}
\theoremstyle{definition}
\newtheorem{theorem}{Theorem}
\newtheorem{lemma}{Lemma}
\newtheorem{definition}{Definition}

\newcommand{\F}{\mathbb F_2}
\newcommand{\E}{\mathbb E}
\newcommand{\rank}{\operatorname{rank}}
\newcommand{\Ber}{\operatorname{Ber}}
\newcommand{\BSC}{\operatorname{BSC}}

\title{MacKay--Neal Codes Achieve Capacity\\under MAP Decoding}
\author{Kenta Kasai}
\date{October 7, 2026}
\begin{document}
\maketitle
\begin{abstract}
Statistical-mechanical analyses predict that MacKay--Neal codes can
achieve channel capacity at fixed degrees. We prove this prediction for
the uncoupled $(\ell,3,3)$ MN socket ensemble for every fixed
integer $\ell\ge4$. For each binary-input memoryless symmetric channel
of capacity strictly greater than $3/\ell$, the ensemble-average block
error probability under maximum a posteriori (MAP) decoding is
$O(\log N/N)$, where $N$ is the transmitted blocklength. The actual
transmitted rate converges to $3/\ell$. The result includes both
punctured and transmitted variables and does not condition the sparse
square matrix on invertibility. The proof establishes an entropy inequality for six bits
subject to even parity by an analytic argument that reduces the domain
by symmetry, restricts its interior stationary points, where both partial derivatives vanish, to the diagonal,
and controls the resulting scalar functions by explicit polynomial bounds. Exact configuration counts then bound the conditional
entropy on the binary symmetric channel. A standard comparison by mutual information
extends the entropy bound to general symmetric channels;
independent output erasures and a minimum-distance estimate for the transmitted code yield the bit
and block error bounds.
\end{abstract}

\section{Introduction}\label{sec:introduction}
Sparse graphical descriptions permit codes with dense effective
constraints while retaining a bounded number of edges per transmitted
bit. MacKay--Neal (MN) codes realize this possibility by puncturing
one class of variables, meaning that those variables are omitted from
transmission~\cite[Sec.~VI-B]{MacKay99}. A rigorous early capacity result uses degrees
that increase as the operating point approaches capacity
\cite{MacKay99}. Statistical-mechanical analyses subsequently
predicted capacity achievement at fixed small degrees
\cite{Murayama00,VanMourik02,TanakaSaad03}.
The question considered here is whether this prediction holds for the
independently matched, uncoupled $(\ell,3,3)$ ensemble, with the
degrees prescribed before the blocklength grows.

\subsection{Main result}
The ensemble has $N$ transmitted variables, $K=3N/\ell$ punctured
variables, and $N$ check nodes. Every check has three sockets of each
edge type; punctured and transmitted variables have socket degrees
$\ell$ and three, respectively. The two types are matched independently
and uniformly, following the multi-edge-type socket construction
\cite[Sec.~7.1]{RichardsonUrbanke08}. Section~\ref{sec:ensemble} gives the exact matrix
definition, including the treatment of parallel edges and singular
matrices. The pair consisting of the punctured and transmitted words
is uniform on the resulting extended linear code.

\begin{theorem}\label{thm:main}
Fix an integer $\ell\ge4$ and set $R=3/\ell$. Let $N$ tend to infinity
through values for which $K=RN$ is integral, and sample the independent
$(\ell,3,3)$ MN socket ensemble of
Definition~\ref{def:ensemble}. For every fixed binary-input memoryless
symmetric channel $W$ with capacity $C(W)>R$, the ensemble-average
conditional entropy of the full extended word is $O(\log N)$.
The ensemble-average transmitted-bit MAP error and extended-word
block-MAP error are both $O(\log N/N)$. The actual transmitted rate
converges to $R$ in $L^1$ and in probability. If $C(W)<R$, the limit
inferior of the ensemble-average transmitted-bit MAP error is at least
$h_2^{-1}(R-C(W))>0$, where $h_2^{-1}$ is the inverse binary entropy
on $[0,1/2]$. Constants in the upper bounds may depend on $W$ and
$\ell$.
\end{theorem}

The theorem makes the fixed-degree prediction rigorous for the
specified ensemble. Its block error assertion concerns the full word,
including punctured variables. The use of a uniform extended-code
prior allows every configuration to be included, without assuming
that the sparse square matrix is invertible. The conclusion holds for each fixed channel and concerns MAP decoding. The equality
case $C(W)=R$ is not covered.

\subsection{Relation to previous results}
The early MN construction uses a sparse source produced, for example,
by reversing arithmetic coding, and treats the source and noise vectors
together in the decoding problem \cite[Sec.~VI-A]{MacKay99}.
Its capacity argument applies the typical-set bounds of Section~II,
with the degree chosen according to the gap to capacity
\cite[Thm.~2 and Sec.~VI-B]{MacKay99}. The present ensemble instead
uses the uniform extended-code prior and prescribed fixed degrees.
Fixed-degree predictions on the
binary symmetric channel (BSC) were subsequently obtained under the assumption of replica symmetry
\cite[Sec.~IV-A]{Murayama00}. The parameters $(K,C,L)$ in that
analysis correspond to $(3,\ell,3)$ here. A later magnetization
enumerator argument uses the Nishimori identity~\cite{Nishimori81},
in its interpretation as inference using the actual prior distribution
and channel law~\cite[Sec.~5.2]{Nishimori01},
together with an additional assumption that the enumerator is independent of the true
noise rate~\cite[Sec.~V, Eqs.~(17)--(19)]{VanMourik02}. An extension
to general symmetric channels obtains the capacity threshold when
the two standard solutions exhaust the locally stable solutions
\cite{TanakaSaad03}.
Theorem~\ref{thm:main} establishes the corresponding capacity
prediction for the specified ensemble by direct finite-length bounds.

Turning these predictions into a rigorous capacity theorem requires
controlling all competing candidates consistent with the received word,
without relying on the assumed structure of the solutions. The Nishimori
condition means that the prior distribution and channel law used for
inference match the actual model, and supplies identities underlying the
earlier analyses \cite{Nishimori81}\cite[Sec.~5.2]{Nishimori01}.
This condition alone does not establish the replica-symmetric ansatz or
show that the solutions considered exhaust those needed for the analysis
\cite{Murayama00,VanMourik02,TanakaSaad03}.

Rigorous capacity achievement at bounded graphical complexity on
general binary-input memoryless symmetric (BMS) channels is already known for Hsu--Anastasopoulos (HA)
codes~\cite{HsuAnastasopoulos10}. The published construction
concatenates an outer Gallager $(j,k)$ low-density parity-check (LDPC) code
with a $(k,k)$ low-density generator-matrix (LDGM) code of rate one, for $j\ge4$ and sufficiently large
even $k$~\cite[Sec.~III-C and Thm.~3]{HsuAnastasopoulos10}.
The sufficient degree bound is uniform over design rates below a
fixed rate upper bound~\cite[Thm.~2 and p.~999]{HsuAnastasopoulos10}.
Consequently, bounded degree and capacity achievement under block
maximum-likelihood (ML) decoding on general BMS channels are established
features of HA codes. The dual HA
description of the present MN family instead has an outer
$(3,\ell)$ LDPC component and an inner $(3,3)$ LDGM component
\cite[Sec.~3.2]{kasaiHA}; these prescribed degrees fall outside
the above hypotheses. The same reference describes the MN capacity
statement as a conjecture supported by replica calculations
\cite[Sec.~3.1]{kasaiHA}.

The role of duality also depends on the channel and the error
criterion. On the binary erasure channel (BEC), let $e_C(\epsilon)$
denote the average extrinsic bit-MAP erasure probability of a length-$n$ binary
linear code $C$. The duality identity
$e_C(\epsilon)+e_{C^\perp}(1-\epsilon)=1$
relates complementary erasure probabilities
\cite[Thm.~4]{Renes18}.
Together with the area identity
$\int_0^1e_C(\epsilon)\,d\epsilon=\dim C/n$
\cite[Thm.~1]{AshikhminKramerTenBrink04},
\cite[Sec.~III]{MeassonMontanariRichardsonUrbanke09},
this transfers
bit-MAP capacity achievement to the dual: if a sequence has rates
converging to $r\in(0,1)$ and vanishing bit-MAP error at every erasure
probability below $1-r$, its dual has vanishing bit-MAP error at
every erasure probability below $r$. Thus duality does provide a
capacity-transfer argument for this channel and error criterion.
For a uniform linear-code input, the ordinary bit-MAP error on the
BEC is $\epsilon e_C(\epsilon)/2$: a bit is ambiguous only if its own
observation is erased, and an ambiguous bit has posterior probabilities
$1/2,1/2$. For such a sequence $C_n$, vanishing bit error below
$1-r$ therefore gives $e_{C_n}(\epsilon)\to0$ there.
The area identity and bounded convergence~\cite[Thm.~1.6.7]{Durrett19} imply
$\int_{1-r}^1[1-e_{C_n}(\epsilon)]\,d\epsilon\to0$.
Since $e_{C_n}$ is nondecreasing, this forces
$e_{C_n}(\epsilon)\to1$ at each $\epsilon>1-r$.
Applying the duality identity at $1-\epsilon>1-r$ proves the stated
bit-MAP conclusion for $C_n^\perp$ at $0<\epsilon<r$.

For general BMS channels, the known channel-duality relation involves
a different output model. In particular, the dual of a BSC has
nonorthogonal pure-state quantum outputs
\cite[Sec.~III-B]{Renes18}.
The classical BMS theorem for HA codes therefore cannot be applied to
that dual channel to obtain the present MN theorem directly.

Block-ML capacity achievement is not preserved by code duality in
general, even on the BEC. To see this, let $D_n$ be a capacity-achieving
sequence of binary linear codes with limiting rate in $(0,1)$ and
append a fixed-zero coordinate,
giving $C_n=\{0\}\times D_n$. This preserves its asymptotic rate and
block error probability. However,
$C_n^\perp=\F\times D_n^\perp$ contains an independent uncoded bit.
Its block-ML error probability is at least $p$ on a BSC of crossover
probability $0<p<1/2$, and at least $\epsilon/2$ on a BEC of erasure
probability $\epsilon>0$. This example concerns block error and is
consistent with the BEC bit-MAP transfer above.

These distinctions explain why the prescribed low-degree family and
the block-error conclusion require additional analysis. Moreover,
MN--HA code duality concerns the transmitted code $C(G)$. Recovering
its punctured variables as well requires the injectivity of $A$,
which is established with high probability in
Lemma~\ref{lem:projected-distance}. The contribution of
Theorem~\ref{thm:main} is the rigorous result for the specified
uncoupled degree family, including extended-word block decoding.

Compound LDGM--LDPC constructions also give source- and channel-coding
optimality and applications to binning~\cite{WainwrightMartinian09}.
The BSC theorem in the published article proves the existence of suitable
finite degrees for each rate--noise pair. Its inner LDGM ensemble
chooses neighbors independently for each output node and has output
degree at least four~\cite[Secs.~III-A--III-B and Thm.~2]{WainwrightMartinian09}.
It therefore does not establish the prescribed doubly regular
degree-three case treated here.

For spatially coupled $(\ell,3,3)$ MN ensembles, the capacity threshold
on the BEC was established using positivity of the density-evolution
potential and threshold saturation~\cite{OkazakiKasai14}.
That result concerns the same degree family, with sufficiently large
coupling width and length, and characterizes iterative decoding through
density evolution. A recent theorem establishes capacity achievement
under the sum-product algorithm (SPA) on BMS
channels~\cite[Thm.~1.1]{Kasai26SC}. The present result bounds
finite-length MAP error probabilities for the independent uncoupled
ensemble. A known
coupled-to-uncoupled entropy interpolation theorem for LDPC ensembles
\cite[Sec.~III-A and Thm.~2]{GiurgiuMacrisUrbanke16} uses a
single edge type. It does not directly apply to checks with three
sockets of each of two separately constrained types. Our proof
controls the uncoupled ensemble through its exact enumerator.

Counting MN candidates requires tracking both punctured and transmitted
variables. Let $a,b$ denote the respective mean fractions of sockets
whose values are one at a check. The exponent in the candidate count
depends on both variables. Exact configuration counts are available
from typical-set analysis \cite[Sec.~VI-B]{MacKay99} and multi-edge-type
LDPC weight enumeration \cite{kasaiMET}; exact socket enumeration
has also been used to study HA/MN distance \cite{Kasai26GV}.
The additional requirement here is to bound this exponent by zero for
every admissible $a,b$ under the capacity condition. Obtaining an exact
counting formula therefore leaves an inequality to be proved before
capacity achievement at the prescribed small degrees follows.

We reduce this requirement to an entropy inequality for a single check.
In the $(\ell,3,3)$ ensemble, each check has three sockets from each
variable class, giving an even-parity constraint on six bits.
Lemma~\ref{lem:universal} bounds their maximum entropy $\Phi(a,b)$ at
the prescribed mean fractions $a,b$ in a form compatible with the
capacity condition. Combined with the configuration count, it bounds
the ensemble-average number of candidates by a polynomial in the
blocklength $N$ for the appropriate range of noise weights.
Changing $\ell$ leaves the check structure unchanged, so the same local
inequality applies to every fixed $\ell\ge4$ in the stated family.
The passage from BSC conditional entropy to general BMS conditional
entropy then uses the established more-capable order, which compares
mutual information for every input distribution \cite{KornerMarton77}:
among equal-capacity BMS channels, the BSC has the smallest mutual
information for every input distribution
\cite[App., Lemma~7.1 and Prop.~7.1]{Sasoglu11}.

Proving the local inequality requires treating unequal mean fractions.
Symmetry under exchanging $a,b$ alone does not imply that a maximum
lies on the diagonal $a=b$. Appendix~\ref{app:parity} studies the
difference between the two sides of the inequality, reduces its domain
using symmetry and boundary estimates, and proves that all interior
stationary points in the reduced domain lie on the diagonal.
Explicit polynomial and power-series bounds then settle the remaining
one-variable problem. Entropy--log-partition duality and the covariance
identities for finite exponential families
\cite[Secs.~3.5--3.6]{WainwrightJordan08},
coefficient-ratio monotonicity for power-series quotients
\cite[Lemma~1]{BiernackiKrzyz55}, and range bounds from Bernstein
coefficients \cite[Sec.~3]{CargoShisha66} are established tools.
The additional parts here are the stationary-point comparison and the
proof of the inequality for the six-bit even-parity constraint.

An entropy bound obtained from candidate counts still needs to be
converted into a decoding-error bound. Small uncertainty about the full
word alone does not imply small bit error: a posterior uniform on the
all-zero and all-one words has entropy one bit but bit-MAP error $1/2$
at every coordinate. We use the capacity gap to erase channel outputs
independently, and bound bit-MAP error through the entropy decrease when
the erased observations are restored in
Lemma~\ref{lem:channel-erasure-slack}. This uses the coordinate-output
comparison underlying the generalized area theorem
\cite[Sec.~II and Thm.~1]{MeassonMontanariRichardsonUrbanke09}.
Lemma~\ref{lem:projected-distance} establishes a minimum-distance bound
for the transmitted code and uniqueness of the punctured word;
Lemma~\ref{lem:block-upgrade} then yields the block-MAP error bound for
the full extended word.

\subsection{Proof outline}
Let $N$ be the transmitted blocklength, let $R=3/\ell$ be the design
rate, and let $C(W)$ be the capacity of the channel $W$. With the
degree $\ell\ge4$ fixed, we aim to show that the ensemble-average
block-MAP error tends to zero as $N$ grows whenever $C(W)>R$.
The proof first counts competing candidates, then bounds conditional
entropy, transmitted-bit error, and extended-word block error in turn.

The candidate count starts from a single parity check. Each check has
three sockets on the punctured-variable side and three on the
transmitted-variable side. Section~\ref{sec:parity} bounds the entropy
of the configurations satisfying this even-parity constraint in terms
of the fractions of ones on the two sides. Varying $\ell$ leaves the
three sockets of each type unchanged, so the same local inequality
serves every degree in the stated family. Appendix~\ref{app:parity}
gives its analytic proof; the coefficient-ratio and Bernstein
estimates cited above are used there to control its maximum.

We first consider the BSC. Write $U$ for the punctured word, $V$ for
the transmitted word, $Y$ for the received word, and $G$ for the
sampled graph. Our target is the remaining uncertainty about the full
word, $H(U,V\mid Y,G)$. The noise-syndrome identity in
Lemma~\ref{lem:channel-syndrome} expresses this uncertainty using
noise and auxiliary-vector pairs with the same syndrome, without
assuming either sparse matrix to be invertible. Within each fixed
Hamming-weight class, the BSC noise is uniform. After also fixing
the syndrome, the corresponding pairs are equiprobable, and their
number determines the conditional entropy through its logarithm.

Combining exact configuration counts~\cite{kasaiMET} with
the local inequality bounds the ensemble-average candidate count
by a polynomial in $N$ for noise weights satisfying the capacity
condition. Its logarithm is $O(\log N)$. Jensen's
inequality~\cite[Thm.~2.6.2]{CoverThomas06} and a bound on the
exponentially rare remaining noise weights~\cite{Hoeffding63} then
give $\E_G H(U,V\mid Y,G)=O(\log N)$ in
Lemma~\ref{lem:bsc-joint-entropy}. The known mutual-information
comparison with an equal-capacity BSC~\cite[Lemma~7.1]{Sasoglu11}
uses the entropy convexity underlying Mrs.\ Gerber's lemma~\cite{WynerZiv73}.
Successive output replacements~\cite[Sec.~I]{ErezOrdentlichShamai24}
extend this comparison to correlated code coordinates. It therefore
transfers this entropy bound to any BMS channel with $C(W)>R$.

Converting this bound to bit error needs an additional argument.
A posterior uniform on the all-zero and all-one words has entropy
one bit, yet every coordinate has MAP error $1/2$. We use the
positive capacity gap $C(W)-R$ to obtain more information: let
$W_\varepsilon$ replace each channel output independently by an
identifiable erasure with probability $\varepsilon>0$, chosen
independently of $N$ so that $C(W_\varepsilon)>R$
still holds. The same $O(\log N)$ entropy bound applies even after
these erasures.

Restore the erased observations one at a time. If the average bit-MAP
error on the original channel remained bounded away from zero,
these observations would decrease the conditional entropy, on
average, by an amount proportional to $N$. After averaging over $G$,
this decrease is bounded by the $O(\log N)$ uncertainty before restoration.
Lemma~\ref{lem:channel-erasure-slack} makes this comparison quantitative,
using the entropy chain rule and mutual-information concavity
\cite[Sec.~2.5 and Thm.~2.7.4]{CoverThomas06}, as in the
coordinate-output comparisons underlying the generalized area
theorem~\cite[Sec.~II and Thm.~1]{MeassonMontanariRichardsonUrbanke09}.
It yields the transmitted-bit bound
$\overline P_{\rm bit}^{\rm MAP}(W)=O(\log N/N)$, where the overbar
denotes averaging over the sampled graph.

To pass to block error, Section~\ref{sec:distance} uses the minimum
distance $d_{\min}$ of the transmitted code. Lemma~\ref{lem:projected-distance}
shows that, outside a set of graphs of probability $O(1/N)$,
$d_{\min}\ge\delta_vN$ for a constant $\delta_v>0$ depending only on
the fixed degree. Select the codeword nearest in Hamming distance to
the vector of bit-MAP decisions. An incorrect selection requires at
least $d_{\min}/2$ wrong bit decisions, so the bit-error bound and
Markov's inequality~\cite[Thm.~1.6.4]{Durrett19} control its block
error without independence assumptions on the decisions. Block-MAP
has no larger error than this decoder. On the same graphs, $A$ is
injective, so $V$ uniquely determines $U$. Lemma~\ref{lem:block-upgrade}
therefore gives $\overline P_{\rm block}^{\rm MAP}(W)=O(\log N/N)$
for the full pair $(U,V)$.

The actual transmitted rate also needs to be established. The entropy
bounds imply that $r_N(G)=\dim C(G)/N$ converges to $R$ in $L^1$,
without a prior rank assumption on the defining matrices. If
$C(W)<R$, the binary Fano inequality~\cite[Sec.~2.10]{CoverThomas06}
gives a strictly positive asymptotic lower bound on the average
bit-MAP error. Together, the two sides identify the capacity threshold
for the actual rate; the equality case $C(W)=R$ remains outside the
statement.

\section{Ensemble and notation}\label{sec:ensemble}
All matrices and vector spaces are over $\F$. Write $|v|$ for Hamming
weight and $d_H$ for Hamming distance
\cite[Sec.~1.4]{RichardsonUrbanke08}. Shannon entropies and channel
capacities use the standard definitions~\cite[Chs.~2 and 7]{CoverThomas06}
and are measured in bits. We use the BMS-channel convention of
\cite[Sec.~4.1]{RichardsonUrbanke08}: channel uses are independent
conditional on the inputs, and exchanging the two input symbols
corresponds to a symmetry of the output law. Its capacity is attained
by the uniform binary input. The natural-log binary entropy is
$h(a)=-a\ln a-(1-a)\ln(1-a)$, with $0\ln0=0$; thus
$h_2(a)=h(a)/\ln2$.

\begin{definition}[MN socket ensemble with two edge types]
\label{def:ensemble}
Fix $\ell\ge4$, $R=3/\ell$, and integers $N,K$ with $K=RN$.
We use the socket-based multi-edge-type construction
\cite[Sec.~7.1]{RichardsonUrbanke08}.
A socket is an edge endpoint assigned to a node before matching.
An edge type specifies corresponding socket classes on the variable
and check sides; only sockets of the same type may be matched.
The uniform matchings for distinct edge types are independent.
The edge-degree vectors are $(\ell,0)$ for punctured variables,
$(0,3)$ for transmitted variables, and $(3,3)$ for every check.
For edge type 1, uniformly match the $\ell K=3N$ sockets
of $K$ variable nodes to three sockets at each of $N$ check nodes.
For edge type 2, independently and uniformly match three sockets
at each of $N$ variable nodes to three more sockets at each check.
Let $A\in\F^{N\times K}$ and $B\in\F^{N\times N}$ be the respective
incidence matrices, with parallel edges added modulo two, and write
$G=(A,B)$. Define the extended and transmitted codes by
\begin{equation}\label{eq:code-definition}
 \mathcal C_{\rm ext}(G)=\ker[A\ B],\qquad
 C(G)=\{v\in\F^N:\text{there exists }u\in\F^K\text{ with }Au+Bv=0\}.
\end{equation}
The pair $(U,V)$ is uniform on $\mathcal C_{\rm ext}(G)$, and $V$ is
sent through $W^{\otimes N}$. The decoder knows $G$.
\end{definition}

For each $v\in C(G)$, the set of corresponding punctured words is a coset
of $\ker A$ and has the same cardinality, so $V$ is uniform on $C(G)$. If $B$ is invertible, the
constraint is the usual MN encoding relation $v=B^{-1}Au$
\cite[Sec.~3.1]{kasaiHA}.
Definition~\ref{def:ensemble} also specifies the prior and transmitted
code when $B$ is singular. In that case, $U$ is uniform on the
feasible subspace $\{u:Au\in\operatorname{im}B\}$, which can be a
proper subspace of $\F^K$: each feasible $u$ has exactly
$2^{\dim\ker B}$ compatible transmitted words. Thus $K$ counts
punctured coordinates, whereas $\dim C(G)$ is the number of
independent information bits in the transmitted code. We prove
below that $\dim C(G)/N$ converges to $K/N$.
Socket degrees refer to the configuration
before reduction modulo two; the weights of the binary matrix rows
or columns can be smaller.

\begin{definition}[Error probabilities and actual rate]
Following the bit and block decision criteria for MAP decoding
\cite[Secs.~1.5 and 4.1.6]{RichardsonUrbanke08}, for fixed $G$, let
\begin{align}
 P_{\rm bit}^{\rm MAP}(G,W)
 &=\frac1N\sum_{i=1}^N
   \min_{\widehat v_i}\Pr\{\widehat v_i(Y,G)\ne V_i\mid G\},
   \label{eq:bit-definition}\\
 P_{\rm block}^{\rm MAP}(G,W)
 &=\min_{\widehat u,\widehat v}
   \Pr\{(\widehat u(Y,G),\widehat v(Y,G))\ne(U,V)\mid G\}.
   \label{eq:block-definition}
\end{align}
An overbar denotes expectation over $G$. The actual transmitted rate
is $r_N(G)=\dim C(G)/N$, the standard code rate applied to the
transmitted projection~\cite[Secs.~1.4 and 1.7]{RichardsonUrbanke08}.
\end{definition}

Conditional entropies at a fixed graph are averaged over graphs only
when an explicit $\E_G$ is written. MAP means maximum a posteriori
decoding under the specified uniform prior and the matched channel
law. Thus block-MAP decoding is also maximum-likelihood decoding of
the extended code, with a fixed convention for ties. The corresponding
success event implies correct recovery of the transmitted word.

\section{An entropy bound for an even-parity check}\label{sec:parity}
A socket is active when the corresponding coordinate of the candidate
vector being counted has value one. The local check polynomial records how many sockets of each type
are active while the total parity is even. Its entropy bound will
control the exact configuration enumerator, in the multi-edge type
counting framework of~\cite{kasaiMET}.
Define
\begin{align}
 F(x,y)
 &=\frac{(1+x)^3(1+y)^3+(1-x)^3(1-y)^3}{2}\notag\\
 &=(1+3x^2)(1+3y^2)+(3x+x^3)(3y+y^3),
 \label{eq:check-polynomial}\\
 \Phi(a,b)
 &=\inf_{x,y>0}\{\ln F(x,y)-3a\ln x-3b\ln y\}.
 \label{eq:phi-definition}
\end{align}
The infimum is understood in the extended real sense, with $\exp(-\infty)=0$. In particular, $\Phi(a,b)=-\infty$ for $|a-b|>2/3$.
For feasible means, $\Phi(a,b)$ is the maximum entropy, in natural-log
units, of an even-parity six-bit vector with marginal means $a$ on
the first three coordinates and $b$ on the last three, by
entropy--log-partition duality
\cite[Sec.~3.6 and Thm.~3.4]{WainwrightJordan08}. The proof
uses the variational expression~\eqref{eq:phi-definition} directly;
infeasible pairs cause no difficulty for the upper bounds.

Each even-parity six-bit vector contributes one monomial to $F$:
the exponents count its ones in the two triples, and the coefficient
counts the vectors with those counts. Thus $F^N$ counts admissible
socket profiles over all $N$ checks. The variational bound
\eqref{eq:phi-definition} controls these coefficients because, for
positive $x,y$, any one coefficient times its monomial is at most
$F(x,y)^N$. Taking the infimum of that upper bound does not require it to be attained. The difficulty in Lemma~\ref{lem:universal}
is to obtain a bound uniform in the two different marginal means;
the separate constraints on the two edge types must both be retained.

\begin{lemma}\label{lem:universal}
For all $a,b\in[0,1]$,
\begin{equation}\label{eq:universal}
 \Phi(a,b)\le 2h(a)+2h(b)+\frac{h(a)h(b)}{\ln2}.
\end{equation}
\end{lemma}

The analytic proof is given in Appendix~\ref{app:parity}. Symmetry and
boundary estimates reduce the problem to a compact square. An
entropy-derivative comparison uses coefficient-ratio monotonicity
\cite[Lemma~1]{BiernackiKrzyz55} to force every interior stationary
point to lie on the diagonal. On the diagonal, a specific substitution and a
second-derivative bound, using positive Bernstein coefficients
\cite[Sec.~3]{CargoShisha66}, handle small marginal means; a power-series
bound handles the remaining interval. The polynomial identities and
rational bounds used in this argument are displayed explicitly.

\begin{lemma}\label{lem:envelope}
For $0\le R\le1$ and $a,b\in[0,1]$,
\begin{equation}\label{eq:rate-envelope}
 (R-3)h(a)-3h(b)+\min\{h(b),(1-R)\ln2\}+\Phi(a,b)\le0.
\end{equation}
\end{lemma}
\begin{proof}
If $\Phi(a,b)=-\infty$, the claim is immediate. Otherwise, for fixed
$a,b$, the left side is piecewise affine in $R$. Its slope
is $h(a)$ for $R\le1-h_2(b)$ and $h(a)-\ln2$ for
$R\ge1-h_2(b)$. Its maximum on $[0,1]$ is attained at
$R=1-h_2(b)$, where it becomes the difference between the two sides
of~\eqref{eq:universal}. Apply Lemma~\ref{lem:universal}.
\end{proof}

\section{Conditional entropy and decoding on BMS channels}
\label{sec:channel}

We first prove a finite-length entropy bound on the BSC, using
noise vectors classified by their weight and the configuration enumerator. A known
more-capable comparison then extends the bound to every BMS channel of
the same capacity. The conversion from joint entropy to bit-MAP error
uses an independent erasure perturbation. Throughout this section,
Shannon entropies and capacities are in bits, whereas $h$ and $\Phi$
use natural logarithms. The ensemble, the variables $(U,V)$, and the
notation $G=(A,B)$ are those of Section~\ref{sec:ensemble}.

The statistical-physics interpretation uses Bayesian inference with
the actual prior distribution and channel law: a valid candidate has weight equal to its channel likelihood,
and its energy is minus the natural logarithm of that likelihood.
Summing these weights gives a partition function, while the posterior
is the corresponding Gibbs distribution~\cite[Sec.~5.2]{Nishimori01}.
Matching the prior and the channel law is the coding-theoretic form
of the Nishimori condition~\cite{Nishimori81}, used in the earlier MN
analyses~\cite[Sec.~IV]{Murayama00}, \cite[Sec.~V]{VanMourik02}. Here we control the posterior uncertainty
by conditioning BSC noise on its weight. After also fixing the
syndrome, all admissible auxiliary-vector and noise pairs have the
same likelihood, so their number determines the conditional entropy
exactly. The noise-syndrome identity below makes this conditioning
precise. This counting argument supplies an
upper bound without requiring a full free-energy formula.

\subsection{A noise-syndrome identity without rank assumptions}

The entropy identity below replaces uncertainty about a codeword by
uncertainty about noise and an auxiliary punctured vector. This makes
the counting problem depend only on the sparse linear map $[A\ B]$.
The vector $U_0$ in the identity is newly sampled; it is not the actual
punctured word $U$. Linearity and the uniform extended-code prior make
the two conditional entropies equal even when $A$ or $B$ is singular.

\begin{lemma}\label{lem:channel-syndrome}
Fix $G$. Let $E\in\F^N$ be the noise of a BSC, let $U_0$ be uniform
on $\F^K$ independently of $E$, and put $T=AU_0+BE$. If $Y=V+E$,
then
\begin{equation}\label{eq:channel-syndrome}
 H(U,V\mid Y,G)=H(U_0,E\mid T,G).
\end{equation}
\end{lemma}
\begin{proof}
Let $D$ have columns forming a basis of $\ker A^T$, and write
$S=D^TBE$. The rows of $D^TB$ span $C(G)^\perp$. Indeed, $z$ annihilates
$C(G)$ precisely when $(0,z)$ is in the row space of $[A\ B]$,
or equivalently when $z=B^Td$ for some $d\in\ker A^T$.
Thus $D^TB$ is a possibly redundant parity-check matrix for the
transmitted code. For invertible $B$, this agrees with the usual
MN--HA dual description~\cite[Sec.~3.2]{kasaiHA}; the derivation above
also covers singular $B$.

The random vector $AU_0$ is uniform on $\operatorname{im}A$. Hence,
conditional on $S$, the vector $T$ is uniform on the corresponding
coset of $\operatorname{im}A$, and
\[
 H(T\mid G)=\rank A+H(S\mid G).
\]
Since $T$ is a function of $(U_0,E)$,
\[
 H(U_0,E\mid T,G)=\dim\ker A+H(E\mid S,G).
\]
For a uniform linear-code input, the posterior noise distribution
given $Y=y$ is its original distribution restricted to the coset with
syndrome $D^TBy$, the standard syndrome-decoding description
\cite[Sec.~1.7.3]{RichardsonUrbanke08}. Consequently
$H(V\mid Y,G)=H(E\mid S,G)$.
For each admissible $V$, the punctured variable $U$ is uniform on the
set of its possible punctured words, of size $2^{\dim\ker A}$, also after observing $Y$. This proves
\eqref{eq:channel-syndrome} without an invertibility assumption on
either sparse matrix.
\end{proof}

\subsection{Finite noise types on the BSC}

The purpose of conditioning on the noise weight is to turn a weighted
posterior into a uniform distribution on a finite set. The size of that
set can then be bounded by exact socket counts. The threshold $p_*$
below is the BSC noise level whose capacity equals $R$; the positive
gap $p_*-p$ makes noise weights above this threshold exponentially rare.

Let $p_*=h_2^{-1}(1-R)$, where $h_2^{-1}$ takes values in $[0,1/2]$.
For $0\le p<p_*$, define
\begin{equation}\label{eq:channel-bound-definition}
 \mathcal B_N(p)=2\log_2(N+1)+\log_2(K+1)+2\log_2(3N+1)
 +(K+N)\exp\{-2N(p_*-p)^2\}.
\end{equation}

\begin{lemma}[BSC conditional-entropy bound]\label{lem:bsc-joint-entropy}
For every $0\le p<p_*$,
\begin{equation}\label{eq:bsc-joint-entropy}
 \E_G H(U,V\mid Y(\BSC(p)),G)\le\mathcal B_N(p).
\end{equation}
In particular, the left side is $O_{p,\ell}(\log N)$.
\end{lemma}
\begin{proof}
Write $W_0=|E|$. Conditional on $W_0=w$, the BSC noise is uniform on
the set of vectors of that weight. To count candidates with the same
syndrome, we use the method for sparse codes~\cite[Sec.~VI-B]{MacKay99}.
For a fixed vector $e$ of that weight, define
\[
 L_G(e,w)=\sum_{u\in\F^K}\sum_{\substack{e'\in\F^N\\|e'|=w}}
               \mathbf 1\{Au+B(e+e')=0\}.
\]
This counts pairs with the same image under the map $(u',e')\mapsto Au'+Be'$
as a given pair $(u_0,e)$, after translating $u'$ by $u_0$.

The coefficient method for multi-edge configuration
enumerators~\cite{kasaiMET} gives the following exact expression:
\begin{equation}\label{eq:bsc-exact-count}
 \E_G L_G(e,w)=
 \sum_{a=0}^{K}\ \sum_{\substack{t\ \text{even}\\
                         0\le t\le2\min(w,N-w)}}
 \frac{\displaystyle\binom Ka\binom w{t/2}\binom{N-w}{t/2}
                  [x^{\ell a}y^{3t}]F(x,y)^N}
      {\displaystyle\binom{3N}{\ell a}\binom{3N}{3t}}.
\end{equation}
To verify the expression directly, select the $a$ nonzero coordinates
of $u$. If $|e'|=|e|=w$ and $|e+e'|=t$, then $t$ is even and
$t/2$ coordinates are changed in each of the support and the
complement of $e$. Under the two independent uniform matchings, the
active row sockets are uniform subsets of sizes $\ell a$ and $3t$.
The coefficient in the numerator counts exactly those subsets for
which each row has even total parity. This reasoning remains valid
when parallel edges cancel modulo two.

The denominator in \eqref{eq:bsc-exact-count} counts all possible
active-socket subsets of the prescribed totals, separately for the
two edge types. The numerator counts the subsets satisfying every
check. Their ratio is therefore the probability that these fixed
supports satisfy the graph constraints. Independence is used between
the two matchings; the check events themselves need not be independent.

Put $\alpha=a/K$, $\beta=t/N$, and $q=w/N$. The exponent in the
count of noise vectors of the prescribed weight is
\[
 J_q(\beta)=q h\!\left(\frac{\beta}{2q}\right)
        +(1-q)h\!\left(\frac{\beta}{2(1-q)}\right),
 \qquad 0\le\beta\le2\min(q,1-q),
\]
with continuous endpoint conventions. Conditioning reduces entropy~\cite[Thm.~2.6.5]{CoverThomas06},
which gives the two bounds
\begin{equation}\label{eq:bsc-shell-envelope}
 J_q(\beta)\le\min\{h(\beta),h(q)\}.
\end{equation}
For the first bound, consider binary variables $E_0,D_0$ with
$\Pr(E_0=1)=q$ and
$\Pr(D_0=1\mid E_0=1)=\beta/(2q)$,
$\Pr(D_0=1\mid E_0=0)=\beta/(2(1-q))$.
Their conditional entropy is $J_q(\beta)$ in natural-log units, and
$\Pr(D_0=1)=\beta$.
For the second bound, $E'_0=E_0+D_0$ also has mean $q$, so
$J_q(\beta)=(\ln2)H(E'_0\mid E_0)\le(\ln2)H(E'_0)=h(q)$.

Use the binary type-class bounds
\cite[Ex.~11.1.3, Eq.~(11.40)]{CoverThomas06}
\[
 \frac{\exp\{n h(k/n)\}}{n+1}\le\binom nk
             \le\exp\{n h(k/n)\},
\]
and coefficient positivity in $F^N$. Each summand in
\eqref{eq:bsc-exact-count} is then at most
\[
 (3N+1)^2\exp\left[N\{(R-3)h(\alpha)-3h(\beta)
                       +J_q(\beta)+\Phi(\alpha,\beta)\}\right].
\]
If $h_2(q)\le1-R$, the exponent is nonpositive by
\eqref{eq:bsc-shell-envelope} and Lemma~\ref{lem:envelope}.
There are at most $(K+1)(N+1)$ summands. Hence, for every such noise weight,
\begin{equation}\label{eq:bsc-fiber-bound}
 \E_G L_G(e,w)\le(K+1)(N+1)(3N+1)^2.
\end{equation}
This coefficient bound holds at finite length.

The passage from counting to entropy is where the order of averaging
matters. For a fixed graph and syndrome, the uncertainty within a noise
class is the logarithm of its number of admissible pairs. We therefore
need an average of $\log L_G$, whereas the enumerator provides an
average of $L_G$. Concavity of the logarithm and Jensen's inequality
\cite[Thm.~2.6.2]{CoverThomas06} give
\[
 \E_G\log_2 L_G(e,w)\le\log_2\E_G L_G(e,w).
\]
The true pair is always admissible, so $L_G(e,w)\ge1$ and the logarithm
is well defined. The polynomial bound \eqref{eq:bsc-fiber-bound} thus
gives logarithmic uncertainty on these classes. Revealing the weight
costs at most $\log_2(N+1)$ bits, and the remaining classes are handled
by their exponentially small probability. The average of the logarithmic
count is bounded by the logarithm of the average count. In statistical
physics these are called the quenched and annealed averages, respectively~\cite[Sec.~8.2.3]{Nishimori01};
equality between those averages is not assumed.

Conditional on $(T,W_0,G)$, the pair $(U_0,E)$ is uniform on the set
of pairs with that syndrome and noise weight. At each fixed graph,
this gives the exact identity
\[
 H(U_0,E\mid T,W_0,G)
 =\E_{U_0,E\mid G}\log_2 L_G(E,W_0).
\]
The count is independent of the actual auxiliary vector $U_0$ because
translation by $U_0$ bijects each syndrome set with the set defining
$L_G$. Since $W_0$ is a function of $E$, the entropy chain rule first
adds $H(W_0\mid T,G)\le H(W_0)$; it does not assume that weight and
syndrome are independent. The entropy identity, the chain rule and maximum-entropy bound for
the revealed weight $W_0$~\cite[Secs.~2.5--2.6]{CoverThomas06},
and Jensen's inequality therefore yield
\begin{align*}
 \E_G H(U,V\mid Y,G)
 &\le H(W_0)+\E_{G,U_0,E}\log_2 L_G(E,W_0)\\
 &\le 2\log_2(N+1)+\log_2(K+1)+2\log_2(3N+1)\\
 &\quad +(K+N)\Pr\{W_0/N>p_*\}.
\end{align*}
Here $W_0/N\le p_*$ implies $h_2(W_0/N)\le1-R$, and on the
exceptional event we used $L_G\le2^{K+N}$. Since
$W_0\sim\operatorname{Bin}(N,p)$, Hoeffding's inequality~\cite{Hoeffding63} bounds the
last probability by $\exp\{-2N(p_*-p)^2\}$, proving the lemma.
\end{proof}

\subsection{The comparison with an equal-capacity BSC}

For a binary-input channel $W$, write $I_q(W)$ for its mutual
information under a $\Ber(q)$ input. A channel is more capable than
another if its mutual information is at least as large for every
input distribution~\cite{KornerMarton77}.

\begin{lemma}[{\cite[Lemma 7.1]{Sasoglu11}}]
\label{lem:bms-more-capable}
If $W$ is BMS with capacity $C$ and $p_C=h_2^{-1}(1-C)$, then
\[
 I_q(W)\ge I_q(\BSC(p_C))\qquad(0\le q\le1).
\]
\end{lemma}
\begin{proof}
A BMS channel is equivalent to a mixture of BSCs whose crossover
$P\in[0,1/2]$ is independent of the input and known to the receiver
\cite[Sec.~4.1.5]{RichardsonUrbanke08}.
Thus $\E h_2(P)=1-C=h_2(p_C)$ and, with
$a\star b=a(1-b)+(1-a)b$,
\[
 I_q(W)=\E h_2(q\star P)-h_2(p_C).
\]
The binary-entropy convexity property underlying Mrs.\ Gerber's lemma~\cite{WynerZiv73} states that
$t\mapsto h_2(q\star h_2^{-1}(t))$ is convex on $[0,1]$.
Applying Jensen's inequality gives the claim. This proof applies to
arbitrary crossover distributions, including channels with continuous
outputs; it does not require finite-mixture quantization.
\end{proof}

This comparison by mutual information extends to repeated memoryless
channel uses. An early channel-comparison
proof uses a chain-rule replacement of coordinate outputs
\cite[App.]{CostaElGamal87}; a modern formulation appears
in~\cite[Sec.~I]{ErezOrdentlichShamai24}. We record the argument with an auxiliary variable to make its application to punctured words explicit.

Code coordinates are generally correlated. The comparison must
therefore hold for every single-bit input probability, not just the uniform
input used to define capacity. In the proof below, conditioning on
other outputs changes the input probability at each coordinate. The single-letter
more-capable inequality still applies to that conditional input probability; memorylessness of
the channel, rather than independence of the code coordinates,
justifies each replacement.

\begin{lemma}\label{lem:channel-tensorization}
Let $X=(X_1,\ldots,X_N)$ have any distribution on $\F^N$, and let
$Q$ be any additional finite random variable. Transmit $X$ through
memoryless channels $W$ and $W'$, with channel noises independent of
$Q$ given $X$. If $W$ is more capable than $W'$, their outputs satisfy
\[
 H(Q,X\mid Y(W))\le H(Q,X\mid Y(W')).
\]
\end{lemma}
\begin{proof}
Couple outputs $Y_i$ of $W$ and $Z_i$ of $W'$ independently given
$X$. In the difference of the mutual informations of successive
hybrid vectors, put
$S_i=(Y_1,\ldots,Y_{i-1},Z_{i+1},\ldots,Z_N)$.
The memoryless property gives
\begin{align*}
 &I(Q,X;Y_1,\ldots,Y_N)-I(Q,X;Z_1,\ldots,Z_N)\\
 &\quad=\sum_{i=1}^N
       \{I(Q,X;Y_i\mid S_i)-I(Q,X;Z_i\mid S_i)\}\\
 &\quad=\sum_{i=1}^N\E_{S_i}
       \{I_{\Pr(X_i=1\mid S_i)}(W)
                  -I_{\Pr(X_i=1\mid S_i)}(W')\}\ge0.
\end{align*}
Subtract from $H(Q,X)$. No independence of the input coordinates is
required. The channel noises here are independent of $Q$ given $X$,
as in the usual experiment of transmitting $X$ through a channel.
\end{proof}

In particular, apply these two lemmas at every fixed graph with
$Q=U$ and $X=V$, and then average over $G$. For any BMS channel $W$
with $C(W)>R$, Lemma~\ref{lem:bsc-joint-entropy} gives
\begin{equation}\label{eq:bms-joint-entropy}
 \E_G H(U,V\mid Y(W),G)
 \le\mathcal B_N\bigl(h_2^{-1}(1-C(W))\bigr)=O_{W,\ell}(\log N).
\end{equation}
The comparison is applied to joint conditional entropy; a direct
ordering of coordinate MAP errors is not needed.

\subsection{Independent output erasures and bit-MAP error}

Sublinear joint entropy alone does not force small bit-MAP error.
For example, a posterior uniform on the all-zero and all-one vectors
has only one bit of entropy, but every coordinate has MAP error
$1/2$. The additional information needed here is that the entropy
bound also holds on a slightly worse channel. Erasing some outputs
creates that channel; restoring them measures how much uncertainty
was associated with the individual bits. The argument that restores erased observations below
uses the entropy chain rule and mutual-information concavity
\cite[Sec.~2.5 and Thm.~2.7.4]{CoverThomas06}; we give the resulting
bound for independent output erasures explicitly. Changing the observed
outputs one at a time uses the approach to conditional entropy in
the generalized area theorem
\cite[Sec.~II and Thm.~1]{MeassonMontanariRichardsonUrbanke09}.

\begin{lemma}\label{lem:channel-erasure-slack}
Let $W$ be BMS with capacity $C>0$, and let $W_\varepsilon$ replace
its output independently by a flagged complete erasure with
probability $0<\varepsilon\le1$. For any distribution of a binary
vector $V$,
\begin{equation}\label{eq:channel-erasure-slack}
 2\varepsilon C N P_{\rm bit}^{\rm MAP}(W)
 \le H(V\mid Y(W_\varepsilon))-H(V\mid Y(W))
 \le H(V\mid Y(W_\varepsilon)).
\end{equation}
\end{lemma}
\begin{proof}
For a BMS channel, concavity of $I_q(W)$, its zero values at $q=0,1$,
and $I_{1/2}(W)=C$ imply
$I_q(W)\ge2C\min(q,1-q)$.
Choose which coordinates are erased and restore the missing channel observations
one at a time. Given the observations available just before revealing
coordinate $i$, let $q$ be the conditional probability that $V_i=1$.
The expected entropy decrease upon this reveal is $I_q(W)$.

Indeed, conditional on $V_i$, the new output at coordinate $i$ is
independent of the other variables and the observations already
available. Thus its information about the entire vector equals its
information about $V_i$ under the current conditional bias $q$.
The quantity $\min(q,1-q)$ is precisely the bit-MAP error before the
reveal, which explains the connection between information gain and
decoding error in this step.
For each fixed set of erased coordinates, the expected MAP error before the reveal is at
least the coordinate MAP error using all original observations,
because the former observation is a function of the latter and
MAP minimizes the decision error~\cite[Sec.~1.5]{RichardsonUrbanke08}. This set
is independent of the input and the original channel outputs.
Sum over erased coordinates and average over this set. Each coordinate
is erased with probability $\varepsilon$, proving the first inequality.
Nonnegativity of conditional entropy proves the second.
\end{proof}

For the ensemble and any fixed $W$ with $C=C(W)>R$, choose
\[
 \varepsilon=\frac{C-R}{2C},\qquad
 C(W_\varepsilon)=\frac{C+R}{2},\qquad
 p_\varepsilon=h_2^{-1}\!\left(1-\frac{C+R}{2}\right)<p_*.
\]
This choice spends half of the capacity gap on the artificial
erasures, leaving $W_\varepsilon$ strictly above rate $R$. Hence the
same logarithmic entropy bound applies to the worse channel. Dividing
the total information gain by the positive factor $2\varepsilon C N$
in Lemma~\ref{lem:channel-erasure-slack} then bounds the bit-MAP error
on the original channel. The strict capacity gap is what makes this
conversion possible.

Apply \eqref{eq:bms-joint-entropy} to $W_\varepsilon$, use
$H(V\mid Y,G)\le H(U,V\mid Y,G)$, and average
Lemma~\ref{lem:channel-erasure-slack} over $G$. This proves the
explicit bound
\begin{equation}\label{eq:channel-bit-bound}
 \overline P_{\rm bit}^{\rm MAP}(W)
 \le\frac{\mathcal B_N(p_\varepsilon)}{(C-R)N}
 =O_{W,\ell}(\log N/N).
\end{equation}
The block-MAP conclusion of Theorem~\ref{thm:main} additionally uses
the minimum-distance argument for the transmitted code in Section~\ref{sec:distance}.

\subsection{Actual transmitted rate and the converse}

The socket counts define $R=K/N$, but dependencies among the matrix
rows can change the actual code dimension. We must control both any
excess dimension of the extended code and the dimension of $\ker A$,
whose vectors have no effect on the transmitted word. The noiseless entropy
bound controls the latter; a channel with capacity just above $R$
controls the former.

Write $\kappa(G)=\dim\ker A$, $k_{\rm ext}(G)=\dim\ker[A\ B]$,
and $k_{\rm tx}(G)=\dim C(G)$. Linear algebra gives
\[
 k_{\rm ext}\ge K,
 \qquad k_{\rm tx}=k_{\rm ext}-\kappa.
\]
On the noiseless channel the full conditional entropy is exactly
$\kappa$. Lemma~\ref{lem:bsc-joint-entropy} at $p=0$ therefore
implies $\E_G\kappa=O(\log N)$.
For any fixed $\eta>0$ with $R+\eta<1$, transmit through a BSC of
capacity $R+\eta$. At each graph, the memoryless-channel capacity bound
\cite[Ch.~7]{CoverThomas06} limits the mutual information in $N$
channel uses to $N(R+\eta)$, so the same lemma yields
\[
 \E_G k_{\rm ext}\le N(R+\eta)+O_{\eta,\ell}(\log N).
\]
Consequently
\[
 \E_G\left|\frac{k_{\rm tx}}N-R\right|
 \le\frac{\E_G(k_{\rm ext}-K)+\E_G\kappa}{N}
 \le\eta+O_{\eta,\ell}(\log N/N).
\]
Let $N\to\infty$ first and then $\eta\downarrow0$. This proves
convergence of the actual transmitted rate to $R$ in $L^1$, hence
also in probability.

Finally, for any fixed BMS channel $W$ and each graph,
\[
 k_{\rm tx}\le NC(W)+H(V\mid Y(W),G)
 \le NC(W)+N h_2\bigl(P_{\rm bit}^{\rm MAP}(G,W)\bigr).
\]
The last step is the binary Fano bound applied coordinatewise, followed
by entropy subadditivity and concavity of $h_2$
\cite[Sec.~2.10 and Thms.~2.6.6 and 2.7.3]{CoverThomas06}.
Average over $G$ and use concavity once more. Rate convergence gives,
when $C(W)<R$,
\[
 \liminf_N\overline P_{\rm bit}^{\rm MAP}(W)
       \ge h_2^{-1}(R-C(W))>0,
\]
as stated in Theorem~\ref{thm:main}.

\section{Minimum distance of the transmitted code and block-MAP decoding}\label{sec:distance}

We next upgrade the transmitted-bit estimate to an estimate for the
entire extended word. The argument uses the same MN socket ensemble
and does not condition on either matrix having full rank. Throughout
this section $\ell\ge4$ is fixed, $R=3/\ell$, and
\[
 \delta_e=\frac{1}{10000\ell^3},\qquad
 \delta_v=\frac{\delta_e}{2},\qquad
 \eta=\frac{\delta_e}{2R}=\frac{1}{60000\ell^2},
\]
\[
 c_\ell=\frac{1-R}{2}\min\{h(\eta),h(3/4)\}>0,
 \qquad Q_N=(K+1)(N+1)(3N+1)^2.
\]

\begin{lemma}\label{lem:projected-distance}
There is an event $\mathcal G_N$, depending only on the sampled graph,
such that
\begin{equation}\label{eq:good-graph-bound}
 \Pr_G(\mathcal G_N^c)
 \le \frac{325\ell^3}{N}+Q_N e^{-c_\ell N}.
\end{equation}
On $\mathcal G_N$, $A$ is injective and every nonzero word of $C(G)$
has Hamming weight greater than $\delta_vN$.
\end{lemma}

\begin{proof}
See Appendix~\ref{app:distance}. The low-weight/linear-weight
first-moment decomposition has a precedent in HA/MN distance analysis
\cite[App.~B]{Kasai26GV}. Here we derive a socket containment bound
for the small-weight region of the two-edge-type ensemble, and use
Lemma~\ref{lem:universal} and the check-polynomial support for the
remaining region.
\end{proof}

A union bound over the transmitted-bit decisions gives only the
upper bound $N P_{\rm bit}$ for their word error probability, which is
insufficient for an $O(\log N/N)$ bit bound. The distance property
in Lemma~\ref{lem:projected-distance} supplies a stronger implication:
selecting a wrong codeword requires errors in a positive fraction of
the transmitted bits. The injectivity of $A$ has a separate role:
it makes each transmitted word determine its punctured variables
uniquely.

\begin{lemma}[From bit-MAP error to block-MAP error]\label{lem:block-upgrade}
For any fixed channel $W$ and the prior specified above, let
$\overline P_{\rm block}^{\rm MAP}(W)$ denote the ensemble-average
block-MAP error probability for the full pair $(U,V)$. Then
\begin{equation}\label{eq:block-upgrade}
 \overline P_{\rm block}^{\rm MAP}(W)
 \le \frac{325\ell^3}{N}+Q_Ne^{-c_\ell N}
      +\frac{2}{\delta_v}\overline P_{\rm bit}^{\rm MAP}(W).
\end{equation}
In particular, the transmitted-bit bound proved above for each fixed
BMS channel with $C(W)>R$ implies
$\overline P_{\rm block}^{\rm MAP}(W)=O(\log N/N)$.
\end{lemma}

\begin{proof}
Fix a graph in $\mathcal G_N$. Let $b(Y)\in\F^N$ be the vector of
coordinatewise MAP decisions, breaking ties in any fixed way. Choose
$\widehat V(Y)\in C(G)$ to minimize Hamming distance to $b(Y)$.
On this graph each transmitted word determines a unique punctured word,
so this defines a decoder for $(U,V)$ as well. The standard half-minimum-distance argument
\cite[Sec.~1.4]{RichardsonUrbanke08} implies that if
$\widehat V\ne V$, then
\[
 \delta_vN<d_H(V,\widehat V)
 \le d_H(V,b)+d_H(b,\widehat V)
 \le2d_H(V,b).
\]

The last inequality follows from the nearest-codeword choice:
the true word $V$ is also an admissible candidate, so
$d_H(b,\widehat V)\le d_H(b,V)$. A block error therefore forces
$d_H(V,b)>\delta_vN/2$. Markov's inequality
\cite[Thm.~1.6.4]{Durrett19} now uses only the
expected number of bit errors, with no independence assumption on
those errors.
It follows that this decoder's block error probability is at most
$2P_{\rm bit}^{\rm MAP}(G,W)/\delta_v$, where the latter bit error
is averaged over the $N$ transmitted coordinates. The block-MAP decoder
has no larger error. On $\mathcal G_N^c$ use the bound one, and then
average over $G$ and apply \eqref{eq:good-graph-bound}. This proves
\eqref{eq:block-upgrade}. The decoder comparison itself does not require
channel symmetry; symmetry is used in the preceding bit-MAP bound.
\end{proof}

\section{Conclusion}\label{sec:conclusion}
The independent uncoupled $(\ell,3,3)$ MN socket ensemble
attains the MAP capacity threshold for every fixed integer $\ell\ge4$.
The proof gives actual-rate convergence and a vanishing block error
probability for each BMS channel whose capacity exceeds the design
rate. It thereby establishes the fixed-degree prediction for this
specific family through finite-length counting and an explicit
parity entropy bound.

The local inequality~\eqref{eq:universal} is independent of both
$\ell$ and the channel. It controls the competing-word exponent
and supplies the negative exponent needed to bound the minimum distance
of the transmitted code.
The BSC-to-BMS step uses an established information ordering; the
additional error-probability arguments are necessary because that
ordering alone does not compare bit or block MAP error.

Further questions include analogous local entropy inequalities for
other MN degrees and sharper finite-length distance and error bounds.
An analysis of efficient decoding would require additional structure,
such as spatial coupling, beyond the uncoupled MAP argument given here.

\appendix

\section{Proof of the entropy bound for an even-parity check}\label{app:parity}
Put $L=\ln2$ and
\[
 \Delta(a,b)=\Phi(a,b)-2h(a)-2h(b)-h(a)h(b)/L.
\]
We prove $\Delta\le0$ by reducing the interior stationary points to the
diagonal and then giving one-variable bounds. All the estimates below
are analytic; the accompanying rational-arithmetic script is only an
optional check of the displayed polynomial identities and constants.

The diagonal reduction is a statement about possible maxima, not
an assumption of equal marginal means. Symmetry alone permits
maxima away from the diagonal. We first settle the boundary and
prove that every interior stationary point of the reduced square
lies on the diagonal. A positive global maximum would then have to
be among those diagonal points, which the scalar estimates exclude.

\paragraph{Symmetry and boundary values.}
Simultaneous complementation gives
$\Phi(1-a,1-b)=\Phi(a,b)$. For $q=1-2a$, $s=1-2b$, and
$D_a=L-h(a)$, substituting $x=a/(1-a)$, $y=b/(1-b)$ gives
\begin{equation}\label{eq:independent}
 \Delta(a,b)\le-D_aD_b/L+\ln(1+q^3s^3).
\end{equation}
Endpoint substitutions are understood as limits. If $qs\le0$, the right side
is nonpositive. Simultaneous complementation therefore reduces the
remaining domain to $[0,1/2]^2$.

The constrained-entropy interpretation of $\Phi$ can be obtained
directly from the finite variational formula for entropy
\cite[Sec.~3.6 and Thm.~3.4]{WainwrightJordan08}. Averaging
over permutations within each triple preserves the two means and
cannot reduce entropy, by entropy concavity
\cite[Thm.~2.7.3]{CoverThomas06}, so the maximizing distribution has marginal
means $a,a,a,b,b,b$. When $a=0$, the first triple vanishes and the last
triple has even parity. Any two of its bits determine the third;
subadditivity~\cite[Thm.~2.6.6]{CoverThomas06} gives $\Phi(0,b)\le2h(b)$ and thus $\Delta(0,b)\le0$.
The same argument applies when $b=0$. When $a=1/2$, substituting $x=1$
gives $\Phi(1/2,b)\le2L+3h(b)$. Equality is attained by choosing the
last three bits independently with mean $b$ and, conditional on their
parity, choosing the first triple uniformly among its four admissible
vectors. Hence $\Delta(1/2,b)=0$, and likewise
$\Delta(a,1/2)=0$.

\paragraph{Interior stationary points lie on the diagonal.}
As an infimum of affine continuous functions of $(a,b)$, $\Phi$ is
upper semicontinuous. It is finite on $[0,1/2]^2$, so $\Delta$ attains
a maximum there. At an interior point, the mean pair $(3a,3b)$ lies
in the interior of the convex hull of the exponent pairs of $F$:
that convex hull contains $[0,2]^2$. The probability of a socket
configuration is proportional to the monomial it contributes to the
check polynomial. The associated finite exponential
family consequently has a unique finite minimizer $(x,y)$, with a
positive-definite covariance matrix. Thus $\Phi$ is differentiable
there and its derivatives are $-3\ln x$ and $-3\ln y$. These are
standard properties of minimal finite exponential families, in which
the pairs of socket counts do not lie on a line
\cite[Secs.~3.5--3.6]{WainwrightJordan08};
the following paragraph checks their hypotheses for this support.

To justify the finite minimizer, put $\theta=\ln x$ and $\zeta=\ln y$,
and let $\mathcal S$ be the exponent-pair support of $F$.
The interior mean pair lies strictly below the maximum support
projection in every unit direction. Compactness of the unit circle
gives a uniform positive gap $c=c(a,b)>0$. Since every nonzero
coefficient of $F$ is at least one,
\[
 \ln F(e^\theta,e^\zeta)-3a\theta-3b\zeta
 \ge\max_{(i,j)\in\mathcal S}
       \{(i-3a)\theta+(j-3b)\zeta\}
 \ge c\sqrt{\theta^2+\zeta^2}.
\]
Thus the objective tends to infinity as $(\theta,\zeta)$ escapes to
infinity, even when the direction varies. Its Hessian, the matrix of second partial derivatives, is the covariance
of the two socket counts under this weighted distribution. The support
spans two dimensions, so that covariance is positive definite.
At the minimizer the derivatives in $\theta,\zeta$ vanish; hence
differentiating the optimum with respect to $a,b$ leaves only the
explicit terms $-3\theta$ and $-3\zeta$.

Write
\[
 r=\frac{1-x}{1+x},\qquad s=\frac{1-y}{1+y},\qquad t=r^3s^3,
 \qquad Q=1-2a,\quad S=1-2b.
\]
The two mean equations give
\[
 Q=\frac{r+r^2s^3}{1+t},\qquad
 S=\frac{s+r^3s^2}{1+t}.
\]
Initially $r,s\in(-1,1)$. The factors $1+rs^3$ and $1+r^3s$ are
positive, so $Q,S>0$ imply $r,s>0$. In particular,
\[
 \operatorname{atanh}Q=\operatorname{atanh}r
                       +\operatorname{atanh}(t/r),\qquad
 \operatorname{atanh}S=\operatorname{atanh}s
                       +\operatorname{atanh}(t/s).
\]
Define
\[
 \delta(u)=\frac{L-h((1-u)/2)}{L},\qquad
 \gamma(u)=\frac{\delta(u)}{u\operatorname{atanh}u},\qquad 0<u<1.
\]
The function $\gamma$ is strictly decreasing by the coefficient-ratio
criterion for power-series quotients~\cite[Lemma~1]{BiernackiKrzyz55}.
Indeed,
\[
 L\delta(u)=\sum_{k\ge1}\frac{u^{2k}}{2k(2k-1)},\qquad
 u\operatorname{atanh}u=\sum_{k\ge1}\frac{u^{2k}}{2k-1},
\]
and the ratios of corresponding coefficients are $1/(2kL)$.
For completeness, the quotient of two positive-coefficient power
series with strictly decreasing coefficient ratios is strictly
decreasing: on differentiating the quotient and pairing terms with
indices $i<j$, each coefficient is a positive factor times
$(a_jb_i-a_ib_j)<0$. Absolute convergence on compact subintervals
justifies these operations.

Differentiating $\Delta$ and using the mean equations yields
\begin{align*}
 \tfrac12\Delta_a
   &=\delta(S)\operatorname{atanh}Q-3\operatorname{atanh}(t/r),\\
 \tfrac12\Delta_b
   &=\delta(Q)\operatorname{atanh}S-3\operatorname{atanh}(t/s).
\end{align*}
Suppose both derivatives vanish and $r>s$. Then
\[
 Q-S=\frac{(r-s)(1-r^2s^2)}{1+t}>0,
 \qquad \frac SQ>\frac sr.
\]
The ratio of the two stationarity equations would give
\[
 \frac SQ\frac{\gamma(S)}{\gamma(Q)}
   =\frac{\operatorname{atanh}(t/r)}
           {\operatorname{atanh}(t/s)}.
\]
The left side is greater than $S/Q>s/r$, whereas the right side is
less than $s/r$, because $\operatorname{atanh}u/u$ is strictly
increasing. This is a contradiction. Interchanging $r$ and $s$
excludes $r<s$. Therefore every interior stationary point satisfies
$r=s$, hence $a=b$. In view of the boundary values, it remains only
to prove $\Delta(a,a)\le0$ for $0\le a\le1/2$.

\paragraph{Rational logarithm bounds.}
The constants used below follow from
\[
 S_6(t)=2\sum_{j=0}^{5}\frac{t^{2j+1}}{2j+1},\qquad
 S_6(t)\le\ln\frac{1+t}{1-t}
 \le S_6(t)+\frac{2t^{13}}{13(1-t^2)}\quad(0<t<1).
\]
Substitution of $t=1/3,1/9,1/5$, respectively, and the identities
$\ln5=2\ln2+\ln(5/4)$ and $\ln6=2\ln2+\ln(3/2)$ give
\begin{equation}\label{eq:analytic-log-bounds}
 \frac{69}{100}<L<\frac{347}{500},\qquad
 \ln5<\frac{10059}{6250}<\frac{161}{100},\qquad
 \ln6>\frac{7167}{4000}.
\end{equation}
These are finite rational inequalities.

\paragraph{The diagonal for $0\le a\le1/6$.}
Substitute $x=y=\sqrt{z}$ with $z=a/[5(1-a)]$.

We need an upper bound on the infimum, so any positive values may be
substituted. The factor five comes from the diagonal stationarity
equation $3a=z\,\partial_z\ln F(\sqrt z,\sqrt z)=15z+O(z^2)$,
whose leading relation is $a=5z+O(z^2)$. The substituted values need not satisfy
that equation exactly. The chosen denominator permits the
remaining correction to be written as the cubic $Q_0$ below.
Since
\[
 F(\sqrt z,\sqrt z)=1+15z+15z^2+z^3,
\]
direct substitution gives
\[
 \Phi(a,a)\le3h(a)+3a\ln5+\ln Q_0(a),\qquad
 Q_0(a)=1-\frac{12}{5}a^2+\frac{176}{125}a^3.
\]
Thus $\Delta(a,a)\le-f(a)$, where
\[
 f(a)=h(a)+h(a)^2/L-3a\ln5-\ln Q_0(a).
\]
We prove that $f$ is concave on $[0,1/6]$ and is nonnegative at its
endpoints.

First put $g(a)=h(a)+h(a)^2/L$, $v=a(1-a)$, and $z=(1-2a)^2$.
For the curvature calculation write $A_0=\operatorname{atanh}(1-2a)$
and $D_0=L-h(a)$. Since $h'(a)=2A_0$ and $h''(a)=-1/v$,
\[
 g''(a)=-\frac{L+2h(a)}{Lv}+\frac{8A_0^2}{L}.
\]
\[
 -Lv\{g''(a)+5\}=\frac{7L}{4}+\frac{5L}{4}z-2D_0-2(1-z)A_0^2.
\]
The entropy-deficit series used above gives $D_0$. The Cauchy product
of the $\operatorname{atanh}$ series gives
\[
 A_0^2=\sum_{k\ge1}\frac{z^k}{k}
                 \sum_{j=1}^{k}\frac1{2j-1}.
\]
Collecting coefficients yields, for $0<a<1/2$,
\begin{equation}\label{eq:analytic-curvature-series}
 -Lv\{g''(a)+5\}
 =\frac{7L}{4}+\left(\frac{5L}{4}-3\right)z
       +\sum_{k\ge2}c_kz^k,
\end{equation}
where
\[
 c_k=\frac{2}{k(k-1)}\sum_{j=1}^{k-1}\frac1{2j-1}
                      -\frac3{k(2k-1)}>0.
\]
Positivity follows already by replacing the sum by one. Using
$L>69/100$ and rounding the coefficients $c_2,\ldots,c_8$ downward
gives the polynomial lower bound
\begin{align*}
 P_8(z)={}&\frac{483}{400}-\frac{171}{80}z+\frac12z^2
             +\frac{244}{1000}z^3+\frac{148}{1000}z^4\\
          &+\frac1{10}z^5+\frac{73}{1000}z^6
             +\frac{56}{1000}z^7+\frac{44}{1000}z^8.
\end{align*}
This polynomial is positive on $[0,1]$. We use the standard range
bound from Bernstein coefficients~\cite[Sec.~3]{CargoShisha66}.
The explicit identity needed here is
\[
 P_8(z)=\sum_{i=0}^8 b_i\binom8i z^i(1-z)^{8-i},
\]
where the numerators of $b_i$, with common denominator $1120000$,
are
\begin{align*}
 (1120000b_0,\ldots,1120000b_4)
   &=(1352400,1053150,773900,519530,297288),\\
 (1120000b_5,\ldots,1120000b_8)
   &=(118790,4940,1610,263200).
\end{align*}
All are positive. The basis weights $\binom8i z^i(1-z)^{8-i}$ are
nonnegative on $[0,1]$ and sum to one by the binomial theorem. Thus
$P_8$ is a weighted average of the positive $b_i$, including at the
endpoints. This proves positivity throughout the interval.
Equation~\eqref{eq:analytic-curvature-series}
therefore implies $g''(a)\le-5$.

Next $Q_0(a)\ge14/15>0$ on $[0,1/6]$, and
\begin{align*}
 5Q_0^2+Q_0''Q_0-(Q_0')^2
  ={}&\frac15+a\left(\frac{1056}{125}-\frac{888}{25}a\right)\\
     &+a^3\left(\frac{17248}{625}-\frac{4224}{125}a^2\right)
       +\frac{357072}{15625}a^4+\frac{30976}{3125}a^6.
\end{align*}
The two parentheses are at least $316/125$ and $49984/1875$,
respectively. Thus $(\ln Q_0)''\ge-5$, and $f''\le0$ follows.

We have $f(0)=0$. At $a=1/6$, write $h_0=9011/20000$ and
$d=203/3375$. The logarithm bounds in
\eqref{eq:analytic-log-bounds} give
\[
 h(1/6)=\ln6-\tfrac56\ln5>h_0,\qquad Q_0(1/6)=1-d.
\]
Since $-\ln(1-d)\ge\sum_{k=1}^4 d^k/k$,
\begin{align*}
 f(1/6)&\ge h_0+\frac{h_0^2}{347/500}-\frac{161}{200}
                                      +\sum_{k=1}^4\frac{d^k}{k}\\
 &=\frac{959977053021773}{11525626687500000000}
   >\frac1{12500}>0.
\end{align*}
Concavity now proves $f(a)\ge0$ throughout $[0,1/6]$.

A concave function lies above the chord joining its endpoint values.
The term $-3a\ln5$ is affine and has no effect on curvature, so the
separate curvature bounds for $g$ and $\ln Q_0$ are sufficient for
this argument. Continuity at zero extends the conclusion to the
closed interval; no derivative at zero is required.

\paragraph{The diagonal for $1/6\le a\le1/2$.}
Set $q=1-2a$, $z=q^2\in[0,4/9]$, and $D=L-h(a)$. The first four
positive terms of the entropy-deficit series give
\[
 D\ge zP_c(z),\qquad
 P_c(z)=\frac12+\frac z{12}+\frac{z^2}{30}+\frac{z^3}{56}.
\]
Also, the alternating logarithm bound gives
\[
 \ln(1+q^6)\le z^2Q_c(z),\qquad
 Q_c(z)=z-\frac{z^4}{2}+\frac{z^7}{3}.
\]
Define $A_c(z)=P_c(z)^2-(347/500)Q_c(z)$. On $[0,4/9]$,
\[
 P_c(z)<\frac{11}{20},\qquad P_c'(z)<\frac18,
 \qquad Q_c'(z)\ge\frac{601}{729}.
\]
Consequently
\[
 A_c'(z)<\frac{11}{80}-\frac12\frac{601}{729}<0,
 \qquad
 A_c(4/9)=\frac{107832863}{70309644300}>0.
\]
Since $Q_c(z)\ge0$ and $L<347/500$, the diagonal case of
\eqref{eq:independent} yields
\[
 \Delta(a,a)\le-\frac{D^2}{L}+\ln(1+q^6)
       \le-\frac{z^2}{L}A_c(z)\le0.
\]
Both diagonal intervals are now covered. If $\Delta$ had a positive
maximum on $[0,1/2]^2$, that maximum would be an interior stationary
point and hence lie on the diagonal, a contradiction. Together with
the initial symmetry reduction, this proves
Lemma~\ref{lem:universal} on the full domain.

\section{Proof of the minimum-distance estimate for the transmitted code}\label{app:distance}
\begin{proof}
We must exclude small transmitted weight even when the punctured
part is large. The proof therefore separates small total weight
$|u|+|v|$ from larger total weight with small $|v|$. The first region
is controlled by socket containment; the second by the entropy
inequality. Including $v=0$ also excludes nonzero invisible pairs
$(u,0)$, which is exactly the injectivity requirement for $A$.

Let $Z_{a,t}(G)$ count the pairs $(u,v)\in\ker[A\ B]$ with
$|u|=a$ and $|v|=t$. The multi-edge coefficient enumeration
\cite{kasaiMET}, applied to independent uniform socket
matchings, gives the exact finite-length identity
\begin{equation}\label{eq:full-weight-enumerator}
 \E_G Z_{a,t}(G)
 =\binom Ka\binom Nt
   \frac{[x^{\ell a}y^{3t}]F(x,y)^N}
        {\binom{3N}{\ell a}\binom{3N}{3t}}.
\end{equation}
All binomial arguments here are integers. The polynomial $F$ records
even total active socket count at each row, including configurations
with parallel edges counted modulo two.

\paragraph{Small extended weights.}
Fix nonzero supports of respective sizes $a,t$, put $d=a+t$, and let
$s=\ell a+3t$ be their total number of active sockets. If the supports
give a kernel word, every occupied row has a positive even number of
active sockets. In particular, $s$ is even. If $s$ is odd the event is
impossible; otherwise write $m=s/2$. For $1\le d\le\delta_eN$,
\[
 \frac{3d}{2}\le m\le\frac{\ell d}{2}<N.
\]
All active sockets of a kernel word must lie in at most $m$ rows.
For any specified set of $m$ rows, independence of the two matchings
and uniformity within each edge type give the containment probability
\[
 \frac{\binom{3m}{\ell a}}{\binom{3N}{\ell a}}
 \frac{\binom{3m}{3t}}{\binom{3N}{3t}}
 \le \left(\frac mN\right)^{\ell a+3t}.
\]
Here a binomial coefficient is zero if its lower argument exceeds its
upper argument. Each nonzero ratio is bounded by the indicated power
by writing it as a product of sampling-without-replacement factors.
Taking the union over sets of $m$ rows therefore yields
\begin{equation}\label{eq:socket-union-bound}
 \Pr_G\{Au+Bv=0\}
 \le \binom Nm\left(\frac mN\right)^{2m}
 \le \left(\frac{em}{N}\right)^m
 \le \left(\frac{e\ell d}{2N}\right)^{3d/2}.
\end{equation}
The last step uses $e\ell d/(2N)\le e\ell\delta_e/2<1$.
There are $\binom{K+N}{d}$ support pairs of total weight $d$.
Consequently, with
\[
 C_\ell=2e\left(\frac{e\ell}{2}\right)^{3/2},\qquad
 b_d=\left[C_\ell\sqrt{\frac dN}\right]^d,
\]
Vandermonde's identity~\cite[Sec.~5.1, Eq.~(5.27)]{GrahamKnuthPatashnik94}
and $K+N\le2N$ give
\begin{equation}\label{eq:small-weight-count}
 \E_G\sum_{a+t=d}Z_{a,t}(G)
 \le \left(\frac{2eN}{d}\right)^d
       \left(\frac{e\ell d}{2N}\right)^{3d/2}
 =b_d.
\end{equation}

For $d\ge2$ with $d+1\le\delta_eN$,
\[
 \frac{b_{d+1}}{b_d}
 =C_\ell\sqrt{\frac{d+1}{N}}
       \left(1+\frac1d\right)^{d/2}
 \le C_\ell\sqrt{e\delta_e}<\frac14,
\]
because
\[
 C_\ell^2e\delta_e=\frac{e^6}{20000}
 <\frac{729}{20000}<\frac1{16}.
\]
Also $b_2=e^5\ell^3/N<243\ell^3/N$. Thus the contribution of
$2\le d\le\delta_eN$ is at most $324\ell^3/N$; the bound remains
valid when this range is empty.

For $d=1$, a transmitted support alone has three active sockets and
cannot be a kernel word. A punctured support alone is also impossible
when $\ell$ is odd. If $\ell$ is even and $1\le\delta_eN$, then
\eqref{eq:socket-union-bound} with $m=\ell/2$ gives a total expectation
at most
\[
 K\left(\frac{e\ell}{2N}\right)^{\ell/2}
 \le \frac{3e^2\ell}{4N}<\frac{7\ell}{N}.
\]
The middle inequality uses $\ell/2\ge2$ and
$e\ell/(2N)<1$. If $1>\delta_eN$, no weight-one term is included.
Markov's inequality~\cite[Thm.~1.6.4]{Durrett19}, followed by
$7\ell\le\ell^3$, now gives
\begin{equation}\label{eq:small-word-probability}
 \Pr_G\{\exists (u,v)\ne0:\ Au+Bv=0,\ |u|+|v|\le\delta_eN\}
 \le\frac{325\ell^3}{N}.
\end{equation}
This proof uses the parity of $\ell a+3t$, and does not require $t$
itself to be even.

\paragraph{Small transmitted weights outside that region.}
For a nonzero coefficient in \eqref{eq:full-weight-enumerator}, put
$\alpha=a/K$ and $\beta=t/N$. The positive-coefficient bound
\[
 [x^{\ell a}y^{3t}]F^N
 \le \inf_{x,y>0}F(x,y)^Nx^{-\ell a}y^{-3t}
 =\exp\{N\Phi(\alpha,\beta)\}
\]
uses $\ell a=3N\alpha$. Applying the binary type-class bounds
\cite[Ex.~11.1.3, Eq.~(11.40)]{CoverThomas06}
\[
 \frac{e^{n h(j/n)}}{n+1}\le\binom nj\le e^{n h(j/n)}
\]
to \eqref{eq:full-weight-enumerator} yields the finite bound
\begin{equation}\label{eq:distance-exponent-bound}
 \E_G Z_{a,t}(G)
 \le(3N+1)^2\exp\{N f_R(\alpha,\beta)\},
\end{equation}
where
\[
 f_R(\alpha,\beta)
 =(R-3)h(\alpha)-2h(\beta)+\Phi(\alpha,\beta).
\]
The coefficient estimate holds at finite length, even when the infimum is not attained.
Lemma~\ref{lem:universal} implies
\begin{equation}\label{eq:distance-exponent-envelope}
 f_R(\alpha,\beta)
 \le [R-1+h_2(\beta)]h(\alpha).
\end{equation}

Suppose $a+t>\delta_eN$ and $t\le\delta_vN$. Then
$R\alpha+\beta>\delta_e$, so $\alpha\ge\eta$.
Every monomial $x^iy^j$ of $F$ has $i-j\le2$. Hence nonzero
coefficients of $F^N$ satisfy
\[
 3N(\alpha-\beta)\le2N,
 \qquad \alpha\le\frac23+\delta_v\le\frac34.
\]
This support restriction prevents $\alpha$ from approaching one.
Moreover, $\delta_v<1/1024$, and the elementary inequality
\[
 h_2(x)\le x\log_2(1/x)+\frac{x}{\ln2}
\]
gives
\[
 h_2(\delta_v)<h_2(1/1024)
 \le\frac{10+1/\ln2}{1024}
 <\frac{12}{1024}<\frac18\le\frac{1-R}{2}.
\]
Here $\ln2>1/2$ and $R\le3/4$ were used. Binary entropy is increasing
on $[0,1/2]$, so \eqref{eq:distance-exponent-envelope} and concavity
of $h$ on $[\eta,3/4]$ imply
\[
 f_R(\alpha,\beta)
 \le-\frac{1-R}{2}h(\alpha)\le-c_\ell.
\]
Zero coefficients contribute nothing. Summing
\eqref{eq:distance-exponent-bound} over at most $(K+1)(N+1)$
pairs of weights therefore bounds the probability of any such word by
$Q_Ne^{-c_\ell N}$.

Let $\mathcal G_N$ be the event that no nonzero extended word has
transmitted weight at most $\delta_vN$. The preceding estimate and
\eqref{eq:small-word-probability} prove
\eqref{eq:good-graph-bound}. On this event, there is no nonzero word
$(u,0)$, so $A$ is injective. Every nonzero $v\in C(G)$ occurs in an extended codeword
$(u,v)$, and therefore has weight greater than $\delta_vN$. This
proves the lemma.
\end{proof}

\section*{Declaration of generative AI use}
This work was developed with substantial assistance from OpenAI Codex,
a generative AI system. Codex was used to explore research ideas and
proof strategies; formulate, develop, and examine mathematical arguments;
assist with literature searches and computational checks; and draft,
revise, and translate the manuscript. The author takes full responsibility
for the mathematical claims, proofs, references, and final text.

\begingroup
\small
\bibliographystyle{IEEEtran}
\bibliography{references}
\endgroup

\end{document}